\documentclass[sigconf, authorversion, nonacm]{acmart}

\usepackage{algorithm}
\usepackage[noend]{algpseudocode}

\DeclareMathOperator\atan{atan}
\DeclareMathOperator\polylog{polylog}

\usepackage{multirow}

\title{Fast real and complex root-finding methods for well-conditioned polynomials}
\author{Guillaume Moroz}
\email{guillaume.moroz@inria.fr}
\affiliation{%
  \institution{Université de Lorraine, CNRS, Inria, LORIA}
  \city{Villers-lès-Nancy}
  \country{France}}

\begin{abstract}
Given a polynomial $p$ of degree $d$ and a bound $\kappa$ on a condition
number of $p$, we present the first root-finding algorithms that return
all its real and complex roots with a number of bit operations
quasi-linear in $d \log^2(\kappa)$. More precisely, several condition
numbers can be defined depending on the norm chosen on the coefficients
of the polynomial. Let $p(x) = \sum_{k=0}^d a_k x^k = \sum_{k=0}^d
\sqrt{\binom d k} b_k x^k$. We call the condition number associated with
a perturbation of the $a_k$ the hyperbolic condition number $\kappa_h$,
and the one associated with a perturbation of the $b_k$ the elliptic
condition number $\kappa_e$. For each of these condition numbers, we
present algorithms that find the real and the complex roots of $p$ in
$O\left(d\log^2(d\kappa)\ \text{polylog}
(\log(d\kappa))\right)$ bit operations.

Our algorithms are well suited for random polynomials since $\kappa_h$
(resp. $\kappa_e$) is bounded by a polynomial in $d$ with high
probability if the $a_k$ (resp. the $b_k$) are independent, centered
Gaussian variables of variance $1$.

\end{abstract}


\begin{document}
\maketitle

\section{Introduction}

The problem of finding all the real or complex solutions of a polynomial
equation $p(z)=0$ has been extensively investigated, both in theory and
in practice. If $p$ is a polynomial of degree $d$ with integer
coefficients of bit size bounded by $\tau$, the state-of-the-art methods
to find the real or complex roots of $p$ require a number of bit
operations in $O(d^2(d+\tau)\polylog(d\tau))$ \cite{Pjsc02, BSSYjsc18}. In the case where
the polynomial is well-conditioned, the best methods in the state of the art
also require at least a quadratic number of bit operations to find its roots. By
well-conditioned, we mean that the variation of the roots of $p$ with
respect to the variation of its coefficients is small
(\cite[chapter 12]{BCSSbook98}, \cite[chapter~14]{Bbook13} and
references therein).

For ill-conditioned polynomials, the distance between two roots can as
small as $2^{-d\tau}$. Pan considered optimal an algorithm that used
$O(d)$ arithmetic operations, where the number of bit operation for each
arithmetic operation is in $O(d\tau)$, and in this sense, he provided a
near-optimal algorithm.  On the other hand, when a polynomial is
well-conditioned, the distance between two roots is not exponentially
small in $d$.

Random polynomials are well-conditioned with a high probability.  More
precisely, let $p(x)$ be a polynomial of degree $d$ where each
of its coefficients is a Gaussian random variable of variance $\binom
d k$. There exist constants $A>1$ and $B>1$ such that 
the so-called \emph{elliptic condition number} (see
Definition~\ref{def:condition}) is lower than $n^A$ with probability
higher than $1-1/n^B$ \cite{CKMWaam12}. A similar result was
proven for the so-called \emph{hyperbolic condition number} when the
variance is $1$ \cite{DNVlms15}.  Moreover, the distribution of the roots of
polynomials with random coefficients is well understood (\cite{EKbams95} and
references therein). Thus it makes sense to provide
algorithms that performs better than the general case for random
polynomials and for well-conditioned polynomial.

Provided that we know a bound $\kappa$ on a condition number of $p$, we
will show that it is indeed possible to find all the roots of $p$ with a
number of operations quasi-linear in $d$ and polynomial in
$\log(\kappa)$.

Even though a condition number was not explicitly used, the analysis of
root-finding methods for well-conditioned polynomials started with Smale
\cite{Sbams81} who studied the probability of failure of the Newton method.
The Newton method is one of the most famous iterative method, that
converges quadratically toward a single root $\zeta$ of $p$ provided
that the initial point is close enough to $\zeta$
(\cite[chapter~8]{BCSSbook98}, \cite[chapter~3]{Dbook06},
\cite[chapter~15]{Bbook13} and references therein). It was later shown
that it is even
possible to construct a set $S_d$ of $d\log^2(d)$ points such that for
all polynomials $p$ and each root $\zeta$ of $p$, there exists a point
in $S_d$ such that the Newton iteration eventually converges toward
$\zeta$ \cite{HSSim01}. Explicit bounds polynomial in the condition
number were derived and improved for multivariate polynomial system of
equations, based notably on homotopy methods
(\cite{CSacm99,CKMWjc08,BPfocm11,Lfocm17} among others).  One drawback of those
approaches is that they require to evaluate $p$ on at least $d$ points,
which leads to a number of arithmetic operations at least quadratic in
$d$.  Some methods based on modified Newton operators, such as the
Weierstrass method (\cite{BGPcma04} and references therein) or the Aberth-Ehrlich method
\cite{Eacm67} were implemented with success, notably in the software
MPSolve \cite{BAna00,BRjcam14}.

For general polynomials, including ill-conditioned ones, fast numerical
factorization is the first approach to provide the state of the art
bound in $O(d^2(d+\tau)\polylog(d\tau))$ \cite{Pjsc02}. However this method is
difficult to implement.

Another family of methods that are efficient in practice are the
subdivision methods. The idea is to subdivide recursively a domain that
contains the roots of $p$ in subdomains, and to reject or accept the
subdomains according to criteria that guarantee that a subdomain
contains one or zero root. For real roots, the criteria that one may use
are notably the Descartes' rule of signs (\cite{RZjcam04} and references
therein),
the Budan's theorem \cite{TEesa06,Ssnc08, ASVnamc08}, or the Sturm's
theorem \cite{BPRbook06} among others. For complex roots,
one may use Pellet's test \cite{BSSYjsc18} or Cauchy's integral theorem
\cite{IPmacis20,IPissac20} among others. Combining subdivision approaches with
Newton iterations allows to match the complexity bound of Pan's algorithm
for real \cite{SMjsc16}.  Subdivision methods are more commonly implemented, notably in
the software ANewDsc \cite{KRSissac16}, SLV \cite{Tacm16}, the package RootFinding in
Maple \cite{maple}, the package real\_roots in sage \cite{sagemath},
Ccluster \cite{IPYicms18} among others.

We can also mention approaches based on the computation of the
eigenvalues of the companion matrix associated to $p$ \cite{Mmn91}.
These approach has the advantage of being numerically stable in many cases
\cite{EMmc95}. These methods are implemented notably in Matlab
\cite{matlab} and numpy \cite{numpy}.

\subsubsection{Contribution}
Focusing on univariate polynomial equations, we develop new algorithms
that are for the first time polynomial in the logarithm of a condition
number, and quasi-linear in the degree. Our approaches work for two
classical condition numbers that we define here for $x$ in the interval
$[0,1]$ and for $z$ in the complex unit disk $D(0,1)$.

Following the theory of
condition number associated to the root-finding problem
\cite[chapter~14~and~16]{Bbook13}, we introduce the following definitions.

\begin{definition}
  \label{def:condition}
  \sloppy Given the polynomial $p(x) = \sum_{k=0}^d a_k x^k = \sum_{k=0}^d \sqrt{\binom d k} b_k x^k$,
  let $f(t) = \cos^d(t) p(\tan(t))$.
  The \emph{real hyperbolic condition number} associated to $p$ is:
  $$\kappa_{h}^{\mathbb R}(p) = \max_{x\in [0,1]} \min\left( \frac{\|a\|_1}{|p(x)|},
  \frac{d\| a\|_1}{|p'(x)|}\right)$$
  The \emph{real elliptic condition number} associated to $p$ is:
  $$\kappa_{e}^{\mathbb R}(p) = \max_{t\in [0,\frac \pi 4]} \min\left( \frac{\|b\|_2}{|f(t)|},
  \frac{\sqrt{d}\|b\|_2}{|f'(t)|}\right)$$

  For $p(z)$ with $z$ in the unit disk, letting $p_{\theta}(x) =
  p(xe^{i\theta})$, we define the \emph{complex hyperbolic} and the
  \emph{complex elliptic} condition numbers as $\kappa_h^{\mathbb C}(p)
  = \max_{\theta \in [0, 2\pi]} \kappa_h^{\mathbb R}(p_\theta)$ and
  $\kappa_e^{\mathbb C}(p) = \max_{\theta \in [0, 2\pi]}
  \kappa_e^{\mathbb R}(p_\theta)$  respectively.
\end{definition}

The justification for the name \emph{hyperbolic} and \emph{elliptic}
comes from the fact that when the $a_k$ are independent,
centered Gaussian variables of variance $1$,
then the density of the root distribution in $[0,1]$  converges to
$1/(\pi(1-t^2))$ when $d$ converges to infinity. Similarly, when the $b_k$ are independent, centered
Gaussian variables of variance $1$, then the root distribution has
density $\sqrt{n}/(\pi(1+t^2))$ \cite{EKbams95}.

Remark that by symmetry of the weights we consider in front of the
coefficients, we can reduce the problem of finding all the roots in
$\mathbb R$ or in $\mathbb C$ to the problem of all finding all the
roots in $[0,1]$ and $\mathbb C$ respectively, through the changes of
variable $x \mapsto -x$ and $x \mapsto 1/x$.

For our algorithms, we consider polynomials with bit-stream
coefficients, where the first $k$ bits can be accessed in $O(k)$ bit
operations. Our output is a list of \emph{approximate zero} as
introduced by Smale \cite{Sbams81}, in the
sense that for any point $z_0$ returned by our algorithm, the sequence
$z_{k+1} = z_k - p(z_k)/p'(z_k)$ converges quadratically toward its
associated root of $p$.  We can now state our main result.

\begin{theorem}
  Let $p(x)$ be a polynomial of degree $d$, with bit-streams
  coefficients.

  There exist two algorithms that finds all its real roots in the
  interval $[0,1]$ in $O(d \log^2(d\kappa)\polylog(\log(d\kappa)))$ with
  $\kappa = \kappa_h^{\mathbb R}(p)$ and $\kappa = \kappa_e^{\mathbb
  R}(p)$ respectively.

  There exist two algorithms that finds all its complex roots
  in the unit disk in $O(d \log^2(d\kappa)\polylog(\log(d\kappa)))$ with
  $\kappa = \kappa_h^{\mathbb C}(p)$ and $\kappa = \kappa_e^{\mathbb
  C}(p)$ respectively.
\end{theorem}

The main idea of our algorithms is to approximate $p$ with a
piecewise polynomial function, where each polynomial has a degree in
$O(\log(d\kappa))$. This is achieved by partitioning the interval
$[0,1]$ and the unit disk following the distribution of the roots. Then
using Kantorovich's theory, we show that a good enough approximation the roots of the
piecewise polynomial is a set of approximated roots associated to all
the roots  of $p$. Our method is summarized in
Algorithm~\ref{alg:rootfinding}.

For the correctness of Algorithm~\ref{alg:rootfinding}, we prove in
key Lemma~\ref{lem:approximation} that if a polynomial $g$ of small
degree is sufficiently close to a series $f$, then the problem of
finding the root of $f$ can be reduced to the problem of finding the roots of $g$. Then in
Section~\ref{sec:hyperbolic:approximation}
and~\ref{sec:elliptic:approximation}, we show that the piecewise
polynomials that we construct in Algorithm~\ref{alg:rootfinding} satisfy
the assumptions of Lemma~\ref{lem:approximation}.

For the bound on the number of bit operations, the main steps that we
need to analyse
in Algorithm~\ref{alg:rootfinding} are Step~$B$ and Step~$C$. In
Step~$C$ we need to solve $\sum_{n=0}^N M_n$ polynomials of degree $P$.
Using a classical algorithm with the state-of-the-art complexity
(\cite[Theorem~2.1.1]{Pjsc02} and \cite{BSSYjsc18}), we can find all the roots in the unit disk of each
polynomial with an error bounded by $2^{-P}$, and with a number of bit
operations in $O(P^3\polylog(P))$. Then, since $P$ is in $O(\log(d\kappa))$ and
the sum of the $M_n$ is in $O(d/\log(d\kappa))$ in all cases (see
Table~\ref{tab:algodata}), we conclude that the bound on the number
of bit operations to perform Step~$C$ is in
$O(d\log^2(d\kappa)\polylog(\log(d\kappa)))$.

In Step~$B$, if we perform the loop as written in
Algorithm~\ref{alg:rootfinding}, this leads to a number of operations
quadratic in $d$. Instead, in Section~\ref{sec:hyperbolic:complexity}
and~\ref{sec:elliptic:complexity}, we show how we can modify Step~$B$
such that the number of bit operations for this step is in
$O(d\log^2(d\kappa)\polylog(\log(d\kappa)))$.

\begin{algorithm}
  \caption{Root-finding algorithm}
  \label{alg:rootfinding}
  \begin{algorithmic}
    \Require \hspace{-0.25cm}\begin{minipage}[t]{7.5cm} \begin{itemize}
          \item[$c$:] list of $d+1$ coefficients
          \item[$t$:] type of the monomial weight ($elliptic$ or $hyperbolic$)
          \item[$\kappa$:] bound on the condition number (see Definition~\ref{def:condition})
        \end{itemize}
      \end{minipage}\\
      \Ensure $result$: \begin{minipage}[t]{6cm}list of the approximate roots of the function\\
        $\begin{cases}
       \sum_{k=0}^d c[k]x^k & \text{ if $t$ is $hyperbolic$}\\
       \sum_{k=0}^d c[k]\sqrt{\binom d k}x^k & \text{ if $t$ is $elliptic$}
     \end{cases}$
   \end{minipage}

    \State
    \State \emph{A. Initialization}
    \State Variables depending on $t$ (see Table~\ref{tab:algodata}):
    \State \hspace{\algorithmicindent}\begin{tabular}{@{}l@{\ }l}
      $v$ & $\gets$ list of $d+1$ monomial functions\\
      $h$ & $\gets$ a scale function\\
      $\gamma$ & $\gets$ list of $N$ real numbers, centers of disks\\
      $\rho$ & $\gets$ list of $N$ real number, radii of disks\\
      $M$ & $\gets$ list of $N$ integers
    \end{tabular}
    \State $P \gets \lceil10\log_2(d\kappa)\rceil$
    \State $w \gets$ list of $P$-th roots of unity
    \For{$0 \leq n \leq N$}
    \State $z[n]$ $\gets$ list of the $M_n$-th roots of unity
    \EndFor
    \State $result \gets$ empty list
    \State
    \State \emph{B. Evaluation}
    \For{$0 \leq n < N$ and $0 \leq p < P$ }
    \For{$0 \leq m < M_n$ }
    \State $e[m,n,p] \gets \sum_{k=0}^d c[k]v[k](\gamma[n]+\rho[n]w[p])z[n,m]^k$
    \State \phantom{$e[m,n,p] \gets$} up to precision $\|c\|2^{-P}$ 
    \EndFor
    \EndFor
    \State
    \State \emph{C. Interpolation and root-finding}
    \For{$0 \leq n < N$}
    \For{$0 \leq m < M_n$}
    \State   $g \gets$ polynomial such that $g(w[p]) = e[m,n,p]$ for all $p$
    \State   \phantom{$g \gets$} with coefficients up to precision
    $\|c\|2^{-P}$
    \State   $s \gets$ roots of $g$ up to precision $\|c\|2^{-P}$
    \For{$0 \leq k < \text{size of $s$}$}
        \State Append $h(\gamma[n]+\rho[n]s[k])z[n,m]$ to $result$
      \EndFor
    \EndFor
    \EndFor
    \State
    \State \Return $result$
  \end{algorithmic}
\end{algorithm}

\begin{table}
  \begin{tabular}{|l|c|c|c|c|}
    \hline
    \multirow{2}{1cm}{Type Domain} & \multicolumn{2}{c|}{Hyperbolic} & \multicolumn{2}{c|}{Elliptic} \\ 
    \cline{2-5}
                                   & $[0,1]$ & $D(0,1)$ & $[0,1]$ & $D(0,1)$ \\
    \hline
    $v_k(x)$ & \multicolumn{2}{c|}{$x^k$} & \multicolumn{2}{c|}{$\binom d k \sin^k(x)\cos^{d-k}(x)$} \\
    \hline
    $h(x)$ & \multicolumn{2}{c|}{$x$} & \multicolumn{2}{c|}{$\tan(x)$} \\
    \hline
    $\tau$ & \multicolumn{4}{c|}{$\log_2(d\kappa)$} \\
    \hline
    $N$ & \multicolumn{2}{c|}{$O(\log(d/\tau))$} &
    \multicolumn{2}{c|}{$O(\sqrt{d/\tau})$} \\
    \hline
    $M_n$ & $1$ & $2^{n+4}$ & 1 & $O( \sqrt{d/\tau})$ \\
    \hline
    $\gamma$ & \multirow{2}{*}{Eq.~\eqref{eq:hr}} &
    \multirow{2}{*}{Eq.~\eqref{eq:hc}} 
             & \multirow{2}{*}{Eq.~\eqref{eq:er}} &
             \multirow{2}{*}{Eq.~\eqref{eq:ec}} \\
    \cline{1-1}
    $\rho$ & & & & \\
    \hline
  \end{tabular}
  \caption{Values for the initialisation of the variables in
  Step~$A$ of Algorithm~\ref{alg:rootfinding}}
  \label{tab:algodata}
\end{table}

First we will prove in Section~\ref{sec:preliminaries} that we can
reduce the root-finding problem to the problem of finding
the roots of a smaller degree polynomial. Then in
Section~\ref{sec:hyperbolic} and~\ref{sec:elliptic}, we will prove the
correctness and bound the complexity of Algorithm~\ref{alg:rootfinding}
for polynomials with small hyperbolic condition number and small
elliptic condition number respectively. Finally in
Section~\ref{sec:extensions}, we will discuss open
questions related to our approach.

\section{Preliminaries}
\label{sec:preliminaries}
\subsection{Notations}

Given a polynomial or an analytic series $f$, we will denote by $f'$ and $f''$ the derivative and the second
derivative of $f$, and by $f^{(k)}$ the $k$-th derivative of $f$. Given
a vector $v$, we will denote by $\|v\|_1$, $\|v\|_2$ and $\|v\|_\infty$
the classical norm $1$, $2$ and infinity of $v$. The transpose of $v$ is
denoted by $v^T$ and its conjugate transpose by $v^H$ and if $w$ is
another vector, $v^H \cdot w$ denotes their scalar product. For a matrix
$A$, we denote by $\|A\|_k$ the induced norm $\sup_{x \neq 0}
\|Ax\|_k/\|x\|_k$.

For a polynomial $p(x) = \sum_{k=0}^d a_k x^k = \sum_{k=0}^d
\sqrt{\binom d k} b_k x^k$, we denote by $\|p\|_1$ the norm $1$ of the
vector $(a_k)$, and by $\|p\|_W$ the norm $2$ of the vector $(b_k)$.

Finally, we will denote by $I$ the interval $[0,1]$, by $U$ the unit
disk, and by $D(\gamma, \rho)$ the complex disk of radius $\rho$
centered at $\gamma$.

\subsection{Roots of approximated polynomial}

Based on Kantorovich's theory, we show that if a polynomial and a
series have coefficients close enough, then the roots of the polynomial
are in the basin of quadratic convergence of the roots of the series.

We state the following theorem for complex roots in the unit disk
$D(0,1) \subset \mathbb C$. Remark that in the case where $f$ and $g$
have real coefficients, it holds for their real roots in the interval
$[0,1] \subset \mathbb R$

\begin{lemma}
  \label{lem:approximation}
  Let $f(x) = \sum_{k=0}^{\infty} f_k x^k$ be
  an analytic series with radius of convergence greater than $1$.
  Assume that there exist $c>0$, $\kappa > 32$, $s > 1$ and an integer $m
  > 2 \log_2(s\kappa^2)$ such that for all point $z$ in the
  unit disk:
  \begin{itemize}
    \item $|f(z)|\leq c/(s\kappa^2)$ implies $|f'(z)| > c/\kappa$,
    \item $|f''(z)| < cs$
    \item for all $k>m$ we have $|f_k| \leq c/2^k$.
  \end{itemize}
  Let $g(x) = \sum_{k=0}^m g_kx^k$ be a polynomial of degree $m$ such that for all
  $0 \leq k \leq m$ we have $|f_k - g_k| \leq c/2^m$.

  Then, for each root $\zeta$ of $f$ in the unit disk, $f$ has no other
  root in $D(\zeta, 1/(2s\kappa))$ and $g$ has a root in the disk
  $D(\zeta, 1/(16s\kappa))$. Moreover, if $g$ has a root $\eta$ in the
  unit disk, then $f$ has a root in the disk $D(\eta, 1/(16s\kappa))$.
\end{lemma}
\begin{proof}
  First, let $\eta$ be a root of $g$ in the unit disk.  Then $|f(\eta)|
  = |f(\eta)-g(\eta)| \leq c\frac {m+1} {2^m} + \frac{c} {2^m} \leq
  c\frac{m+2}{2^m}$ using the bounds on the difference of the
  coefficients of $f$ and $g$. 
  In particular, with the lower bound on $m$, we have $m \geq
  \log_2(\kappa) + m/2$ and $m \geq 20$ since $s\kappa^2 \geq
  2^{10}$. This implies that $|f(\eta) \leq c(m+2)/2^m \leq c/(s\kappa^2) \cdot
  (m+2)/2^{(m/2)} \leq c/(32s\kappa^2)$.  This implies that $|f'(\eta)|
  \geq c/\kappa$. In turn, we have $\beta = |f(\eta)/f'(\eta)| \leq
  1/(32s\kappa)$, and $K = \max_{z \in U}(|f''(z)/f'(\eta)|) \leq
  s\kappa$. Thus, $2\beta K \leq 1/16 \leq 1$. Using Kantorovich's theory
  \cite[Theorem~88]{Dbook06}, this
  ensures that $f$ has a root in $D(\eta, 2\beta)$ which implies that
  $f$ has a root in the disk $D(\eta, 1/(16s\kappa))$. Moreover, using
  Kantorovich's theory again \cite[Theorem~88]{Dbook06}, since $2K \leq
  2s\kappa$, this implies that $\zeta$ is the only
  root of $f$ in the disk $D(\zeta, 1/(2s\kappa))$.

  Reciprocally, let $\zeta$ be a root of $f$ in the unit disk. Then
  $|g(\zeta)| = |g(\zeta)-f(\zeta)| \leq c\frac{m+2}{2^m}$ using the
  bounds on the difference of the coefficients of $f$ and $g$.
  Similarly $|g'(\zeta) - f'(\zeta)| \leq c\frac{m(m+1)}{2^{m+1}} +
  c\frac{m+2}{2^{m}} \leq c\frac{(m+2)^2}{2^{m+1}}$. And for all $z \in
  U$ we have also $|g''(z) - f''(z)| \leq c\binom {m+2}{3}/2^{m-1} +
  c(m^2+3m+4)/2^m \leq c \frac{(m+3)^3}{3 \cdot 2^m}$.

  This implies that:
  \begin{align*}
    |g(\zeta)| &\leq c\frac{m+2}{2^m}\\
    |g'(\zeta)| &\geq c/\kappa - c\frac{(m+2)^2/2}{2^m}\\
    |g''(z)| &\leq cs + c\frac{(m+3)^3/3}{2^m}
  \end{align*}
  In particular, with the lower bound on $m$, we have $m \geq
  \log_2(\kappa) + m/2$ and $m \geq 20$, which implies $|g'(\zeta)|
  \geq c/(2\kappa)$ and $|g''(z)| \leq c(s + 1/40)$. Such that
  $\beta := |g(\zeta)/g'(\zeta)| \leq (m+2) \kappa / 2^{m-1}$ and $K :=
  \max_{z \in U}(|g''(z)/g'(\zeta)|) \leq 2(s+1/40)\kappa \leq 4.1/2 s \kappa$.

  Let $r = 1/(4.1s\kappa)$. Using Kantorovich's theory
  \cite[Theorem~85]{Dbook06} this implies that $\zeta$ is the unique root
  of $f$ in the disk $D(\zeta, r)$.

  Moreover, $\beta/(2r) \leq 4.1s\kappa^2(m+2)/2^{m}\leq 4.1(m+2)/2^{m/2}
  \leq 1/8$ for $m \geq 19$, which is the case since $s\kappa^2 \geq
  2^{10}$. In this case, using Kantorovich's theory
  again \cite[Theorem~88]{Dbook06}, $2 \beta K \leq \beta/r \leq 1$ ensures that $g$
  has a root in $D(\zeta, 2\beta)$. Moreover, since $\beta/(2r) \leq 1/8$,
  this implies that $2\beta \leq r/2$ and $g$ has a root 
  $\eta$ in the disk
  $D(\zeta, r/2)$. In particular, $\zeta$ is the only root of $f$ in the
  disk $D(\eta, r/2)$, which implies that $\zeta \in D(\eta,
  1/(16s\kappa))$ and thus $\eta \in D(\zeta, 1/(16s\kappa))$.
\end{proof}

\section{Hyperbolic case}
\label{sec:hyperbolic}

In this section we consider the polynomial $p(x) = \sum_{k=0}^d a_k
x^k$, over the interval $[0,1]$ and over the complex unit disk $D(0,1)$.

\subsection{Bound on the coefficients}

For a complex number $\gamma$ and a real number $\rho$, we define the
polynomial $p_{\gamma,\rho}(x) = p(\gamma + \rho x)$.

\begin{lemma}
  \label{lem:hyperbolic:taylor}
  Let $\rho > 0$  be a real and $\gamma$ a complex number in $D(0,1)$ such that
  either $2\rho \leq 1 - |\gamma|$, or $\rho \leq \tau/(2ed)$.
  Let $c_k$ be the coefficients of $x^k$ in $p_{\gamma,\rho}(x)$.
For all $k > \tau$:
\begin{align*}
  |c_{k}| &\leq \frac 1 {2^k} \|p\|_1
\end{align*}
\end{lemma}
\begin{proof}
  We distinguish 2 cases. 
  For the case where $2\rho \leq 1-|\gamma|$, the coefficient of $x^k$
  in $p_{\gamma,\rho}$ is
    $c_{k} = \sum_{i=k}^d a_i \binom i k |\gamma|^{i-k}\rho^k
            \leq \frac 1 {2^k} \sum_{i=k}^d a_i \binom i k
            |\gamma|^{i-k} (1-|\gamma|)^k \leq \|p\|_1/2^k$.

  \sloppy Then, for the case $\rho \leq \tau/(2ed)$ we have
  $c_{k} = \sum_{i=k}^d a_i \binom i k |\gamma|^{i-k}\rho^k \leq \rho^k
  \binom d k \|p\|_1$.
  Using the inequality $k! \geq \sqrt{2\pi k} \left(k/e\right)^k$ we
  get $\binom d k \leq \frac 1 {\sqrt{2\pi k}} \left({ed}/k\right)^k$.
  Which implies, for all $k>\tau$ that $c_k \leq \|p\|_1 \left(\tau/(2k)
  \right)^k\leq \|p\|_1/2^k$.
\end{proof}

\subsection{Piecewise polynomials over $[0,1]$}

Let $(\gamma_n)_{n=0}^{N-1}$ and $(\rho_n)_{n=0}^{N-1}$ be the sequences:
\begin{equation}
\begin{split}
  \gamma_n &= 1 - \frac 2 3 \frac 1 {3^n} \\
  \rho_n &= \begin{cases} \frac 1 3 \frac 1 {3^n}& \text{if $0\leq n < N-1$}\\
                      1-\gamma_n & \text{if $n=N-1$}\\
        \end{cases}
\end{split}
\label{eq:hr}
\end{equation}
where $N = \lceil \log_3\left(\frac {4ed} {\tau}\right) \rceil$ is chosen such
that $\rho_{N-1} \leq \frac \tau {2ed}$.
Remark that the union of the intervals $[\gamma_n-\rho_n, \gamma_n+\rho_n]$ is the interval $[0,
1]$.

Finally, for $0 \leq n < N$, Lemma~\ref{lem:hyperbolic:taylor} implies
that the coefficients $c_k$ polynomial $p(\gamma_n + t \rho_n)$ satisfy
$c_k \leq \|p\|_1/2^k$ for all $k > \tau$.

\subsection{Piecewise polynomials over $D(0,1)$}

In this section, we define a set of disks that cover the disks unit disk
of radius $1$ centered at $0$, while their centers and radii still
satisfy the conditions of Lemma~\ref{lem:hyperbolic:taylor}.

Let $(r_n)_{n=0}^N$ be the sequence:
\begin{align*}
  r_n = \begin{cases} 1 - \frac 1 {2^n} & \text{if $0 \leq n < N$}\\
                      1                         & \text{if $n=N$}
        \end{cases}
\end{align*}
Then for $0 \leq n < N$, let $\gamma_n = \frac 1 2 (r_n + r_{n+1})$ and
$\rho_n = \frac 3 4 (r_{n+1} - r_n)$, such that $(\gamma_n)_{n=0}^{N-1}$
and $(r_n)_{n=0}^{N-1}$ are the sequences:
\begin{equation}
\begin{split}
  \gamma_n &= \begin{cases} 1 - \frac 3 4 \frac 1 {2^n} & \text{if $0 \leq n \leq N-2$}\\
                       1 - \frac 1 2 \frac 1 {2^n} & \text{if $n=N-1$}
         \end{cases}\\
  \rho_n &= \begin{cases} \frac 3 8 \frac 1 {2^n} & \text{if $0 \leq n \leq N-2$}\\
                       \frac 3 4 \frac 1 {2^n} & \text{if $n=N-1$}\\
         \end{cases}
\end{split}
\label{eq:hc}
\end{equation}
where $N = \lceil \log_2\left(\frac {3ed} {\tau}\right) \rceil$ is chosen such
that $\gamma_{N-1} \leq \frac \tau {2ed}$.

Let $\omega_n = \frac {2\pi} {2^{\min(n+4, N+2)}}$. The following lemma
shows that the union of the disks $D(\gamma_ne^{i m \omega_n}, \rho_n)$ for $0 \leq n < N$ and $0 \leq m < 2^{\min(n+4, N+2)}$ contains the disk $D(0, 1)$.

\begin{lemma}
  The disk of center $\gamma_n$ and radius $\rho_n$ covers a sector of
  angle $\frac {2\pi} {2^{\min(n+4, N+2)}}$ of the ring between the concentric circles centered at
  $0$ of radii $r_n$ and $r_{n+1}$.
\end{lemma}
\begin{proof}
  Consider the ring between the circles of radii $r_n$ and
  $r_{n+1}$ and let $\alpha_n$ be the angle of the sector covered by
  the disk $D(\gamma_n, \rho_n)$. Using classical trigonometric formula we have
  $r_n^2 = \gamma_n^2 + r_{n+1}^2 - 2\gamma_n r_{n+1}\cos\left(\frac
  {\alpha_n} 2\right)$, and we also have $\rho_n = \frac 3 2(r_{n+1} -
  \gamma_n)$, which implies:
  \begin{align*}
    \sin\left(\frac{\alpha_n} 2\right) &= \sqrt{1 - \frac {(\gamma_n^2 +
    r_{n+1}^2 - \rho_n^2)^2}{4\gamma_n^2r_{n+1}^2}}\\
                &= \sqrt{\frac {2 \rho_n^2(\gamma_n^2 + r_{n+1}^2)  - (r_{n+1}^2
                - \gamma_n^2)^2}{4\gamma_n^2r_{n+1}^2}}\\
                &= \sqrt{\begin{aligned} & \frac 9 8 \frac {(r_{n+1} -
                    \gamma_n)^2}{r_{n+1}^2}(1 + \frac {r_{n+1}^2}
                    {\gamma_n^2}) \\
                                         & - \frac 1 4 (\frac
                                     {(r_{n+1}-\gamma_n)(r_{n+1}+\gamma_n)}
                                 {r_{n+1} \gamma_n})^2 \end{aligned}}\\
                &= \frac {r_{n+1} - \gamma_n}{2r_{n+1}}\sqrt{\frac 9 2 (1 +
                \frac {r_{n+1}^2} {\gamma_n^2}) - (\frac
              {r_{n+1}}{\gamma_n} + 1)^2}
  \end{align*}
  A variation analysis shows that for $1 \leq x \leq 2$,
  the expression $\frac 9 2 (1 + x^2) - (1+x)^2$ is greater or equal to
  $5$. Moreover, $\frac{r_{n+1}-\gamma_n}{2r_{n+1}}$ is greater than $\frac 1
  {2^{n+3}}$ if $n < N-1$ and greater than $\frac 1 {2^{N+1}}$ if
  $n=N-1$, such
  that:
  \begin{align*}
    \frac {\alpha_n} 2 \geq \sin\left(\frac{\alpha_n} 2\right) & \geq \frac {\sqrt{5}}{2^{\min(n+3, N+1)}} \geq \frac {2\pi}{2^{\min(n+5,N+3)}}
  \end{align*}
\end{proof}

Remark that like for the real case, for $0 \leq n < N$,
Lemma~\ref{lem:hyperbolic:taylor} implies that the coefficients $c_k$
polynomial $p(\gamma_n + t \rho_n)$ satisfy $c_k \leq \|p\|_1/2^k$ for
all $k > \tau$.

\subsection{Approximation properties}
\label{sec:hyperbolic:approximation}
We show in this section that the polynomials computed in
Algorithm~\ref{alg:rootfinding} computes the correct approximate roots of
$p$. For that, we show that with the parameters chosen in the algorithm,
Lemma~\ref{lem:approximation} applies correctly and thus, the approximate
truncated polynomials that we use return the correct roots.
We focus on the complex case. The real case can be
proven with similar arguments.

Let $\tau \geq 6e$ be a real number, let $\gamma \in D(0,1)$ and $\rho > 0$
such that either $2\rho \leq 1 - |\gamma|$ or $\rho \leq \tau/(2ed)$.
Moreover, assume that $\rho \geq \tau/(6ed)$. Denote by
$p_{\gamma,\rho}(z)$ the polynomial $p(\gamma + \rho z)$ and denote by
$c_k$ its coefficients.  For $z \in U$, $|p'_{\gamma,\rho}(z)| = \rho
|p'(\gamma + \rho z)|$ and $p''_{\gamma,\rho}(z) = \rho^2f''(\gamma +
\rho z)$.

\begin{lemma}
With $c = \|p\|_1$, $s = \tau d 2^\tau$, $\kappa = \kappa_h(p)$
and $m = \lceil 2 \log_2(s\kappa^2) \rceil$,
$p_{\gamma,\rho}$ satisfies all the assumptions of
Lemma~\ref{lem:approximation}.
\end{lemma}
\begin{proof}
First, by definition of $\kappa_h$, if $|p_{\gamma,\rho}(z) = p(\gamma + \rho z)| \leq
c/\kappa$, then $|p'(\gamma + \rho z)| \geq cd/\kappa$, which
implies $|p'_{\gamma,\rho}(z)| \geq c\tau/(6e\kappa) \geq c/\kappa$.

For the second derivative of $p''_{\gamma,\rho}(z)$, remark first that $|\gamma
+ \rho z| \leq |\gamma| + \rho \leq 1+\tau/(2ed)$. Thus, for all $0
\leq k \leq d$, we have $|\gamma + \rho z| \leq (1+\tau/(2ed))^d \leq
e^{\tau/(2e)} \leq 2^\tau$. Thus, $|p''_{\gamma,\rho}(z)| \leq \rho^2
|p''(\gamma + \rho z)| \leq \|p\|_1 \tau d 2^\tau$.

Finally, for $k > \tau, |c_k| \leq \|p\|_1/2^k$ using
Lemma~\ref{lem:hyperbolic:taylor}.
\end{proof}

\subsection{Complexity to evaluate $p$}
\label{sec:hyperbolic:complexity}

We focus now on the complexity of Step~$B$ in
Algorithm~\ref{alg:rootfinding}. We modify the algorithm to
be able to bound correctly the number of bit operations of this step.

The following lemma first shows how to evaluate quickly the points near the
unit circle.

\begin{lemma}
  \label{lem:hyperbolic:evaluation1}
  Let $\tau>0$ be a real number and $N>0$ be an
  integers such that $N \leq 64 d/\tau$. Given a complex number $z$
  such that $|z| \leq 1 + \tau/d$ and an integer $m>0$, it is possible to
  compute the $N$ values $p(ze^{i2\pi k/N})$ for $0 \leq k < N$ with an
  absolute error lower than $\|a\|_1/2^m$ and with a number of
  bit operations in $O(d/\tau(\tau + m + \log(d))^2 \cdot \polylog(m+\log(d)))$.
\end{lemma}
\begin{proof}
  For any $0 \leq k \leq d$, remark that $|z^k| \leq (1+\log(2)\tau/d)^d
  \leq 2^\tau$.
 Using fast algorithms, we
can compute in quasi-linear time the first $n$ digits of the result of arithmetic operations
\cite{GGbook13}. 
Thus, we can evaluate the first $\tau + m + 2\log(dm)$ digits of $z^k$ with a
number of bit operations in $T(\tau,m,d) =
O((\tau+m+\log(d))\polylog(m+\log(d)))$. This
allows us notably to evaluate $p(z)$ with an error lower than $\|a\|_1/2^m$ in
$O(dT(\tau,m,d))$ bit operations.

For $N$ in $O(d/\tau)$, we want to evaluate $p$
on the $N$-th roots of unity. We start by computing the
polynomial $q(X) = p(X) \mod X^N - 1$ of degree $N-1$ with a number of
bit operations in $O(d T(\tau,d,m))$. Then we can use the fast Fourier
transform to evaluate $q(e^{i2\pi k/N})$ for $0 \leq k < N$ in
$O(d/\tau\log(d)T(\tau,m,d))$ bit operations.
\end{proof}

Then we show how the points $z$ in the disks $D(\gamma_n, \rho_n)$ that satisfy $|z| \leq 1 -
1/2^n$ and can be evaluated more efficiently.

\begin{lemma}
  \label{lem:hyperbolic:evaluation2}
  Let $n$ be a positive integer and $N>0$ be an
  integers such that $N \leq 2^{n+4}$. Given a complex number $z$
  such that $|z| \leq 1 - 1/2^n$ and an integer $m>0$, it is possible to
  compute the $N$ values $p(ze^{i2\pi k/N})$ for $0 \leq k < N$ with an
  absolute error lower than $\|p\|_1/2^m$ and with a number of
  bit operations in $O(2^n(m + \log(n))^2 \cdot \polylog(m))$.
\end{lemma}
\begin{proof}
  First, remark that $|z^k| \leq e^{-k/2^n}$. In particular, for $k >
  \log(2)m2^n$, we have $|z^k| \leq 1/2^{m}$. 
  Let $\widetilde p$ be the polynomial $p$ truncated to the degree
  $d_n = \lceil \log(2)m2^n\rceil$. Each $z^k$ for $k \leq d_n$ can be computed with an error
  less than $1/2^m$ and with a number of bit operations in $T(m) =
  O(m\polylog(m))$. The polynomial $q(X) = \widetilde p(X) \mod X^N -
  1$ can be computed with a number of bit operations in $O(m2^nT(m))$. Then, using the fast Fourier transform
  approach, we can compute $q(ze^{i2\pi k/N})$ with a number
  of bit operations in $O(2^n\log(N)T(m))$.
\end{proof}

Thus, combining Lemma~\ref{lem:hyperbolic:evaluation1}
and~\ref{lem:hyperbolic:evaluation2}, with $m$ and $\tau$ in
$O(\log(d\kappa))$, $N$ in $O(\log(d/\tau))$, $P$ in $O(\tau)$, $M_n =
1/2^{n+4}$, we can compute $p((\gamma_n + \rho_n e^{i2\pi k/P}) e^{i2\pi
\ell/M_n})$ for all $0 \leq n \leq N$, $0 \leq \ell < M_n$ and $0 \leq k
< P$ with a number of bit operations in $O(d \log^2(d\kappa)
\polylog(d\kappa))$.

\section{Elliptic case}
\label{sec:elliptic}

In this section, we consider the polynomial $p(x) = \sum_{k=0}^d
\sqrt{\binom d k} b_k x^k$ and we define the function $f(x) = \cos^d(x)
p(\tan(x))$.  Remark that for $x \in [0, \pi/4]$, the function $\tan(x)$
is a bijection between the roots of $f$ in $[0,\pi/4]$ and the roots of
$p$ in $[0,1]$.  Moreover, let $v_k(x) = \sqrt{\binom d k}
\sin^k(x)\cos^{d-k}(x)$ and let $v(x)$ be the vector map $(v_0(x),
\ldots, v_d(x))^T$.  Using the notations of the introduction, the
function $f$ can be rewritten:
\begin{equation*}
  f(x) = b^H \cdot v(x)
\end{equation*}

Letting $\alpha_k = \sqrt{k(d+1-k)}$, Edelman and Kostlan \cite{EKbams95}
observed that the derivative of $v$ satisfies the equation $v'(x) = A
\cdot v(x)$, where $A$ is the anti-symmetric linear matrix:
\begin{equation*}
  A =
  \begin{pmatrix}
    0 & -\alpha_1 & & & &\\
    \alpha_1 & 0 & -\alpha_2 & & & \\
            & \alpha_2 & 0 & -\alpha_3 & &\\
      &          & \ddots & \ddots & \ddots  &\\
      &          & & \alpha_{d-1} & 0 & -\alpha_d\\
      &          & &            &\alpha_d & 0\\
  \end{pmatrix}
\end{equation*}

This leads to the following relation:
\begin{equation*}
  v(x) = e^{x A}v(0)
\end{equation*}

As a corollary, for any point $z \in
D(0,1)$:

\begin{equation*}
  f^{(k)}(z) = b \cdot e^{z A} A^k v(0)
\end{equation*}

\subsection{Bound on the derivatives of $f$}

For any real $x$, observe that the matrix $e^{xA}$ is orthogonal because
$A$ is antisymmetric. This allows to prove the following lemma.

\begin{lemma}
  \label{lem:elliptic:taylor}
  For any real $x$:
  $$\left|\frac{f^{(k)}(x)}{k!}\right| \leq \|b\|_2
  \left(\max(4,2\sqrt{ed/k})\right)^k$$ 
\end{lemma}
\begin{proof}
First using norm inequality, we have:
\begin{equation*}
  \left|f^{(k)}(x)\right| \leq \|b\|_2 \|A^k v(0)\|_2
\end{equation*}
For a positive integer $r$, let $A_r$ be the matrix $A$ where all the
entries of indices $(m,n)$ with $m \geq r+2$ or $n \geq r+2$ are replaced by $0$.

Since $A$ is a tridiagonal matrix, and since $v(0) = ( 1, 0, \cdots,
0)^T$, we can deduce by induction that:
\begin{equation*}
  A^k v(0) = A_k\cdots A_1 v(0)
\end{equation*}

Let $h=\left\lfloor \frac {d+1} 2\right\rfloor$.
For $r\leq h$, we can bound the norm of $A_r$ by:
\begin{align*}
  \|A_r\|_2 &\leq \sqrt{\|A\|_1\|A\|_\infty}\\
            &\leq \sqrt{(r-1)(d+1-(r-1))}+\sqrt{r(d+1-r)}\\
            &\leq 2\sqrt{r(d+1-r)}\\
\end{align*}

For $r > h$, we have $\|A_r\|_2 \leq d+1$. This allows us to
deduce that:
\begin{equation*}
  \frac 1 {k!} \|A^k v(0)\|_2 \leq \begin{cases}
    2^k \sqrt{\binom d k} &\text{ if } k \leq \frac {d+1} 2\\
    2^h\sqrt{\binom d h}
    (d+1)^{k-h}\frac{h!}{k!} &\text{ otherwise}
  \end{cases}
\end{equation*}

Using the inequality $k! \geq \sqrt{2\pi k} \left(\frac k e\right)^k$ we
get $\binom d k \leq \frac 1 {\sqrt{2\pi k}} \left(\frac {ed} k\right)^k$.
Moreover for $k > (d+1)/2$, observe that $\sqrt{\binom d h} \leq \sqrt{2^d} \leq 2^h \leq
2^k$, and
$(d+1)^{k-h}\frac{h!}{k!} \leq 2^{k-h}$,
such that:
\begin{equation*}
  \frac 1 {k!} \|A^k v(0)\|_2 \leq \begin{cases}
    \left(2\sqrt{\frac {ed} k}\right)^k
     &\text{ if } k \leq \frac {d+1} 2\\
     4^k
     &\text{ otherwise}
  \end{cases}
\end{equation*}
\end{proof}

\subsection{Piecewise polynomials over $[0,1]$}

Using the bound on the derivative of $f$ shown in the previous section
we defined a sequence of disks $D(\gamma_n, \rho_n)$ that covers the real segment $[0, \pi/4]$
such that the series $f(\gamma_n + \rho_nz)$ has the absolute value of its coefficients $f_k$
decreasing exponentially for $k$ large enough.

For a real $\tau$, let $N = \left\lceil \frac \pi 2 \sqrt{\frac{ed}{\tau}}
\right\rceil$, and for $0 \leq n < N$ let
$\gamma_n$ and $\rho_n$ defined by:
\begin{equation}
  \begin{split}
  \gamma_n &= (2n+1)\frac 1 4 \sqrt{\frac \tau {ed}}\\
  \rho_n &= \frac 1 4 \sqrt{\frac \tau {ed}}
  \end{split}
  \label{eq:er}
\end{equation}

It is easy to check that the union of the corresponding disks cover the
segment $[0, \pi/4]$. The properties of the series $f(\gamma_n, \rho_n)$
will be analysed in Section~\ref{sec:elliptic:approximation}.

\subsection{Piecewise polynomials over $D(0,1)$}

For the complex case, we need to define a sequence of disks $D_n =
D(\gamma_n, \rho_n)$ such that the union of the sets $\tan(D_n)$ covers an
angular sector of $D(0,1)$ big enough. For that, we prove the following lemma.

\begin{lemma}
Let $0 \leq \theta \leq \pi/4$. If a set of complex
disks $D_1, \ldots, D_k \subset \mathbb C$ covers the band $B_\theta$ of points $z$ with $|Im(z)| \leq \theta$ and $0
\leq Re(z) \leq \pi/4$, then $\tan(D_1), \ldots, \tan(D_k)$ covers the
angular sector $A_\theta$ of the unit disk between the angle $-\theta$ and
$\theta$.
\end{lemma}
\begin{proof}
Using the integral expression of the function $\atan$, remark that
$\atan(a+ib) = \atan(a) + \int_{z=a}^{a+ib} \frac 1 {1+z^2} dz$. In
particular, as long as $b \leq a$, we have $Re(z^2) \geq 0$, such that
$|\frac 1 {1+z^2}| \leq 1$, which allows us to conclude that
$|\atan(a+ib) - \atan(a)| \leq |b|$. Moreover, if $a \geq 0$ and
$a^2+b^2\leq 1$ then, $0\leq Re(\atan(a+ib)) \leq \pi/4$. Thus, for any
point $z \in A_\theta$, we have $\atan(z) \in B_\theta$, such that
$A_\theta \subset \tan(B_\theta)$.
\end{proof}

Thus, we can cover a band of width $\frac 1 4 \sqrt{\frac \tau {2ed}}$ with
$N= \left\lceil \pi \sqrt{\frac{ed}{\tau}} \right\rceil$ disks
defined for $0 \leq n < N$ by:
\begin{equation}
  \begin{split}
  \gamma_n &= n\frac 1 4 \sqrt{\frac \tau {ed}}\\
  \rho_n &= \frac 1 4 \sqrt{\frac \tau {ed}}
  \end{split}
  \label{eq:ec}
\end{equation}

This allows us to cover the angular sector of radius $\theta \geq \frac
1 4 \sqrt{\frac \tau {2ed}}$, and the number of sectors needed to cover
the unit disk is $M = \lceil \pi/\theta \rceil \leq \lceil 4\pi \sqrt{\frac
  {2ed} \tau} \rceil$.

\subsection{Approximation properties}
\label{sec:elliptic:approximation}

We focus in this section on the complex case. The real case can be
proven with similar arguments.

For a complex number $\gamma$ and a real number $\rho$, denote by $f_{\gamma,\rho}(z)$ the series
$f(\gamma + \rho z)$ and denote by $f_k$ its coefficients.  For $z \in
U$, $|f'_{\gamma,\rho}(z)| = \rho |f'(\gamma + \rho z)|$ and
$f''_{\gamma,\rho}(z) = \rho^2f''(\gamma + \rho z)$.

\begin{lemma}
  \label{lem:elliptic:approximation}
  Let $0 \leq \gamma \leq \pi/4$ and $\rho = \frac 1 4
  \sqrt{\tau/(ed)}$ be two real numbers. With $c = \|b\|_2$, $s = \tau d 2^\tau$, $\kappa = \kappa_e(f)$
and $m = \lceil 2 \log_2(s\kappa^2) \rceil$,
$f_{\gamma,\rho}$ satisfies all the assumptions of
Lemma~\ref{lem:approximation}.
\end{lemma}
\begin{proof}
First, by definition of $\kappa_e$, if $|f_{\gamma,\rho}(z) = f(\gamma + \rho z)| \leq
c/\kappa$, then $|f'(\gamma + \rho z)| \geq c\sqrt{d}/\kappa$, which
implies $|f'_{\gamma,\rho}(z)| \geq c\sqrt{\tau}/\kappa \geq c/\kappa$.

For the second derivative of $f''_{\gamma,\rho}(z)$, remark that $f''(z)
= b \cdot A^2 e^{iAz} \cdot v(0) = b \cdot A^2 v(z)$. 
Remark that $v(a+ib) = e^{aA}v(ib)$ and $\|v(ib)\|_2 =
(\cosh^2(b) + \sinh^2(b))^{d/2} = \cosh^{d/2}(2b)$. Using the
inequality $\cosh(x) \leq e^{\frac{x^2}2}$ this leads to $\|v(a+ib)\|_2
\leq e^{db^2}$. With $|b| \leq \rho$, this leads to
$|f''_{\gamma,\rho}(z)| \leq \rho^2\|b\|_2\|A\|_2^2e^{\log(2)\tau}$. Moreover,
Eldeman and Kostlan showed that $iA$ is similar to the Kac matrix
\cite[\S 4.3]{EKbams95}, and the absolute value of its eigenvalues is lower
or equal to $d$, such that $\|A\|_2 \leq d$ and $\rho^2\|A\|_2^2 \leq \tau d$.

Finally, for $k > \tau, |f_k| \leq \|b\|_2/2^k$ using
Lemma~\ref{lem:elliptic:taylor}.
\end{proof}

Thus, the two sequences of disks defined in
Equations~\eqref{eq:er} and~\eqref{eq:ec} cover the interval $[0,1]$ and
the unit disk $D(0,1)$ respectively, and they satisfy the conditions of
Lemma~\ref{lem:elliptic:approximation}.

\subsection{Complexity to evaluate $f$}
\label{sec:elliptic:complexity}

In the elliptic case, evaluating a sequence of points in Step~$B$ of
Algorithm~\ref{alg:rootfinding} naively would be done roughly in a
$O(d^2)$, or in $O(d^{3/2})$ operations if we use the fast Fourier
transforms. In both cases, this would exceed our complexity bound. The
main idea in this section is to remark that if we are interested in
computing an approximate value of the function $f = \sum_{k=0}^d b_k
v_k(z)$ up to $\|b\|_2/2^m$ for a given integer $m$, then we can
truncate $f$ to use a support of size in $O(\sqrt{dm})$.

\begin{lemma}
  \label{lem:truncate}
  Given a function $f(z) = \sum_{k=0}^d b_kv_k(z)$ and an integer $m>0$,
  there exists $0 \leq \ell \leq u \leq d$ such that $|u-\ell| \leq
  4\sqrt{dm}$ and $\left|f(z) - \sum_{k=0}^u b_k v_k(z)\right| \leq
  \|b\|_2 2^{-m-1}$.
\end{lemma}
\begin{proof}
  \sloppy Let $m$ be an integer, $x$ be a real between $0$ and $1$ and $\ell = \max(0,\lfloor
  xd-\sqrt{2\log(2)d(m+1)} \rfloor)$ and $u = \min(d, \lceil
  xd+\sqrt{2\log(2)d(m+1)} \rceil)$. Let $I$ be the union of the indices
  $0,\ldots,l$ and $u,\ldots,d$.
  Using the Hoeffding inequality, we have $\sum_{k \in I} \binom d k
  x^k(1-x)^{d-k} \leq 2 \cdot 2^{-4(m+1)}$. In particular, this implies that
  $|\sum_{k \in I} b_k\sqrt{\binom d k}
  \sin(z)^k\cos(z)^{d-k}|^2 \leq \|b\|_ 2^2 \sum_{k=0}^\ell \binom d k
  |\sin^{2k}(z) \cos^{2(d-k)}(z)|$. If $z = a+ib$, we have $|\cos^2(z)| +
  |\sin^2(z)| = \cosh(2b)$. Thus, letting $x = |\sin^2(z)|/\cosh(2b)$,
  we can use the Hoeffding inequality and deduce:
    $$|\sum_{k \in I} b_k\sqrt{\binom d k} \sin(z)^k\cos(z)^{d-k}| \leq
    \|b\|_2\cosh^{d/2}(2b)2^{-2(m+1)+1}$$

  Moreover, comparing the coefficients of the Taylor expansion at $0$ of
  $\cosh(x)$ and
  $\exp(x^2/2)$, remark that $\cosh^{d/2}(2b) \leq e^{db^2}$. In
  particular, if $|b| \leq \sqrt{\log(2)m/d}$ that implies $\cosh^{d/2}(2b) \leq
  2^m$. This allows us to conclude that $|f(z) - \sum_{k=l}^u
  b_k\sin^k(z)\cos^{d-k}(z)| \leq \|b\|_2 2^{-m-1}$.
\end{proof}

Truncating $f$ can also be used to evaluate it efficiently on a set of
roots of unity using fast Fourier transform, as required for
Step~$B$ of Algorithm~\ref{alg:rootfinding}.

\begin{lemma}
  \label{lem:elliptic:evaluation}
  Let $\tau>0$ be a real number and $M>0$ be an
  integers such that $M \leq 4\pi\sqrt{2ed/\tau}$. Given a complex number $z$ such that $|Im(z)| \leq \sqrt{\log(2)\tau/d}$, let $f_z(\omega) = \sum_{k=0}^d
  b_kv_k(z)\omega^k$. Given an integer $m>0$, it is possible to
  compute the $M$ values $f_z(e^{i2\pi k/M})$ for $0 \leq k < M$ with an
  absolute error lower than $\|b\|_2/2^m$ with a number of
  bit operations in $O(\sqrt{d/\tau}(\tau + m + \log(d))^2 \cdot \polylog(m+\log(d)))$.
\end{lemma}
\begin{proof}
For any $0 \leq k \leq d$, and $z=a+ib$, remark that $|v_k(a+ib)| \leq |\cos^2(z)| +
|\sin^2(z)| = \cosh^{d/2}(2b) \leq 2^\tau$. Using fast algorithms, we
can compute in quasi-linear time the first $n$ digits of the result of arithmetic operations
\cite{GGbook13}. Moreover, using methods such
as the FEE method \cite{Kpit91}, we can also evaluate
trigonometric, exponential and factorial functions in quasi-linear time.
Thus, we can evaluate the first $\tau + m + 2\log(dm)$ digits of $v_k(z)$ with a
number of bit operations in $T(\tau,m,d) =
O((\tau+m+\log(d))\polylog(m+\log(d)))$. Using Lemma~\ref{lem:truncate}, this
allows us notably to evaluate $f_z(1)$ with an error lower than $\|b\|_2/2^m$ in
$O(\sqrt{dm}T(\tau,m,d))$ bit operations after
truncating $f_z$.

\sloppy Let $g_z(\omega)$ be the polynomial of degree $4\sqrt{dm}$ approximating
$f_z(\omega)$. For $M$ in $O(\sqrt{d/\tau})$, we want to evaluate $g_z$
on the $M$-th roots of unity. We start by computing the
polynomial $h_z = g_z \mod X^M - 1$ of degree $M-1$ with a number of
bit operations in $O(\sqrt{dm}T(\tau,d,m))$. Then we can use the fast Fourier
transform to evaluate $h_z(e^{i2\pi k/M})$ for $0 \leq k < M$ in
$O(\sqrt{d/\tau}\log(d)T(\tau,m,d))$ bit operations. Thus the total
number of bit operations is in $O(\sqrt{d\tau}(\sqrt{m\tau} +
\log(d))T(\tau, m, d))$ and $O(\sqrt{m\tau}) = O(\tau + m)$.
\end{proof}

\sloppy Finally, Lemma~\ref{lem:elliptic:evaluation} with $m$ and $\tau$ in
$O(\log(d\kappa))$, $N$ in $O(\sqrt{d/\tau})$, $P$ in $O(\tau)$, $M_n$
in $O(\sqrt{d/\tau})$ for all $n$, we can compute all the values in
Step~$B$ of Algorithm~\ref{alg:rootfinding} with a number of bit
operations in $O(d \log^2(d\kappa) \polylog(d\kappa))$.

\section{Extensions and open questions}
\label{sec:extensions}

\subsection{Flat polynomials}

A third natural family of polynomials is of the form $p(x) =
\sum_{k=0}^d \frac 1 {\sqrt{k!}} c_kx^k$. When $d$ converges to
infinity, the density of its roots distribution converges to $1/\pi$.
Thus, we can define a so-called \emph{flat condition number} as follow.

\begin{definition}
  Given the polynomial $p(x) = \sum_{k=0}^d \sqrt{\frac 1 {k!}} c_k x^k$,
  let $f(x) = p(x)e^{-x^2/2}$.
  The real flat condition number associated to $p$ is:
  $$\kappa_{f}^{\mathbb R}(p) = \max_{x\in \mathbb R} \min\left(
    \frac{\|c\|_2}{|f(x)|},
  \frac{\|c\|_2}{|f'(x)|}\right)$$

  For $p(z)$ with $z$ in the unit disk, letting $p_{\theta}(x) =
  p(xe^{i\theta})$, we let $\kappa_f^{\mathbb C}(p) = \max_{\theta \in
  [0, 2\pi]} \kappa_f^{\mathbb R}(p_\theta)$ be the complex flat
  condition number associated to $p$.
\end{definition}

Considering this new condition number, several natural questions occur.
First, remark that the density of the distribution of the roots of flat
polynomials is close to the density of the distribution of the
eigenvalues of random matrices. Whereas it was shown that the
expectation of the hyperbolic
condition number of the characteristic polynomial of complex standard
Gaussian matrices of size $n$ is in $2^{\Omega(n)}$, it
would be interesting to analyse the flat condition number of such
characteristic polynomials.

From an algorithmic point of view, remark that in the flat case,
considering the vector of function $v(x) = (e^{-x^2/2}, x e^{-x^2/2},
\ldots, x^k e^{-x^2/2}, \ldots)$, the derivation of $v$ is an
anti-symmetric operator, as for the elliptic case. Thus, after dealing
with boundary conditions, we should be able to derive an algorithm that find the
roots of such polynomials with number of bit operations linear in $d$
and polynomial in the logarithm of $\kappa_f$. Such an algorithm might
be well suited to find the roots of characteristic polynomials.

\subsection{Multivariate polynomial systems}

As mentioned in introduction, the current bound on the number of
operation to find the roots of a multivariate polynomial systems of
equations is currently polynomial in its degree and in its condition
number. It would be nice to generalize for the multivariate case the
tools that we used and developed for the univariate case.

In particular, Lemma~\ref{lem:approximation} is based on Kantorovich's
theory, where all theorems are valid for multivariate systems. Moreover,
the distribution of the roots is also well described for the
multivariate case \cite{EKbams95}. Combining those results as
we did for the univariate case could help to improve the
state-of-the-art bounds on the problem of finding the roots of
multivariate polynomial systems.

\subsection{Bound on the condition number}
Although the algorithm we present here is quasi-linear in the degree of
the polynomials and polynomial in the logarithm of its condition number,
it requires that a bound on the condition number is given as input. If
the bound given as input is to low, the results might be wrong.

On the other hand, in the complex case, using the piecewise polynomial
approximation and the error bound that we compute with our algorithm, we
can use Kantorovich's theory to check if each root that we compute is
indeed associated to a root of the original polynomial. If we get $n$
distinct roots, then our result has been validated with a number of
bit operations in $O(d\log^2(d\kappa)\polylog(\log(d\kappa)))$.

\bibliographystyle{ACM-Reference-Format}
\bibliography{bibliography}

\end{document}